\documentclass[reqno]{amsart}

\usepackage{booktabs} 
\usepackage{IEEEtrantools}
\usepackage{amssymb,latexsym,amsfonts,amsmath}
\usepackage{graphicx}
\usepackage{paralist}
\usepackage{cite}
\usepackage[
singlelinecheck=false 
]{caption}
\usepackage{eso-pic}
\usepackage{epsfig}
\usepackage{mathtools}
\usepackage{tikz}
\usepackage{tikz}
\usepackage{bm}
\usepackage{dsfont}
\usepackage[bb=boondox]{mathalfa}
\usepackage{subfig}
\usepackage{epstopdf}
\usepackage{graphicx}
\usepackage{epsfig}
\usepackage{mathtools}
\usepackage{amsfonts}
\usepackage{lineno}
\usepackage{tikz}
\usetikzlibrary{positioning,shapes}
\usetikzlibrary{arrows}
\usepackage[justification=centering]{caption}
\allowdisplaybreaks
\newtheorem{theorem}{Theorem}
\newtheorem{definition}[theorem]{Definition}

\newtheorem{assumption}[theorem]{Assumption}
\newtheorem{proposition}[theorem]{Proposition}
\newtheorem{remark}[theorem]{Remark}

\newcommand{\TF}{\mathcal{F}}
\newcommand{\norm}[1]{\left\Vert #1\right\Vert}

\newcommand{\op}{{\mathcal H}}
\def\field#1{\mathbb #1}%
\def\R{\field{R}}%
\def\N{\field{N}}%
\def\Z{\field{Z}}%

\newcommand{\SF}{\mathcal{S}}

\newcommand{\Let}{:=}
\newcommand{\KK}{\mathcal{K}_{\infty}}
\newcommand{\intcc}[1]{\ensuremath{{\left[#1\right]}}}

\usepackage{color}

\usepackage{xspace}

\usepackage{enumerate}
\definecolor{myco}{rgb}{0.55, 0.0, 0.63}

\renewcommand{\emptyset}{{\varnothing}}
\begin{document}
	
	\title[Compositional Synthesis of Symbolic Models for Networks of Switched Systems]	{Compositional Synthesis of Symbolic Models for Networks of Switched Systems$^*$}
	
	\thanks{$^*$ This works has been accepted for publication in the IEEE Control Systems Letters
		(L-CSS)}
	\author{Abdalla Swikir$^1$}
	\author{Majid Zamani$^{2,3}$}
	\address{$^1$Hybrid Control Systems Group, Technical University of Munich, Germany.}
		\email{abdalla.swikir@tum.de}
	\address{$^2$Computer Science Department, University of Colorado Boulder, USA.}
	\email{majid.zamani@colorado.edu}
	\address{$^3$Computer Science Department, Ludwig Maximilian University of Munich, Germany.}
	\maketitle


\begin{abstract}                          
	In this paper, we provide a compositional methodology for constructing symbolic models for networks of discrete-time switched systems. We first define a notion of so-called augmented-storage functions to relate switched subsystems and their symbolic models. Then we show that if some dissipativity type conditions are satisfied, 
one can establish a notion of so-called alternating simulation function as a relation between a network of symbolic models and that of switched subsystems. The alternating simulation function provides an upper bound for the mismatch between the output behavior of the interconnection of switched subsystems and that of their
symbolic models. Moreover, we provide an approach to construct symbolic models for discrete-time switched subsystems under some assumptions ensuring incremental passivity of each mode of switched subsystems.
Finally, we illustrate the effectiveness of our results through two examples.
\end{abstract}

	\section{Introduction}
The notion of symbolic models (a.k.a. finite abstractions) plays an important role
in the control of hybrid systems (see \cite{Tabu} and the references
therein). Symbolic models allow us to use automata-theoretic methods \cite{MalerPnueliSifakis95} to design controllers for hybrid systems with respect to logic specifications such as those expressed as linear temporal logic (LTL) formulae \cite{Katoen}. Symbolic models are established for incrementally stable switched systems, a class of hybrid systems \cite{liberzon}, by providing approximate bisimulation relations between them \cite{Girard,Gossler,SAOUD,Corronc}. However, as the complexity of constructing symbolic models grows exponentially in the number of state variables in the concrete system, the approaches proposed in \cite{Girard,Gossler,SAOUD} limit the application of symbolic models to only low-dimensional switched systems. 
Although the result in \cite{Corronc} provides a state-space discretization-free approach for computing symbolic models of incrementally stable switched systems, this approach is still monolithic and reduces the computational complexity only for switched systems with few modes, see \cite[Section IV(D)]{Corronc}.

Motivated by the above limitation, in this work we aim at proposing a compositional framework for constructing symbolic models for interconnected switched systems. To do so, we first $i)$ partition the overall concrete switched system into a number of concrete switched subsystems and construct symbolic models of them individually; $ii)$ then establish a compositional scheme that allows us to construct a symbolic models of the overall network using those individual ones.

The compositional framework based on a divide-and-conquer scheme \cite{Michael} is not new. Several results have already introduced compositional techniques for constructing symbolic models of networks of control subsystems. The results in \cite{Tazaki2008,7403879,Majumdar,arxiv,swikir} provide techniques to approximate networks of control subsystems by networks of symbolic models by assuming some stability property of the concrete subsystems. Other compositional approaches provide techniques to design symbolic models of concrete networks without requiring any stability property or condition on the gains of subsystems \cite{meyer,omar,Kim}. 
However, none of the aforementioned results in \cite{Tazaki2008,7403879,Majumdar,arxiv,swikir,meyer,omar,Kim} provide a compositional framework for constructing symbolic models for interconnected switched systems. 

In this paper, we provide a compositional methodology for the construction of symbolic models of interconnected switched systems based on dissipativity theory \cite{murat}. We first define a notion of so-called augmented-storage functions to relate switched subsystems and their symbolic models. Then, by leveraging dissipativity-type compositional conditions, we construct a notion of so-called alternating simulation functions as a relation between the interconnection of switched subsystems and that of their symbolic models. This alternating simulation function allows one to determine quantitatively the mismatch between the output behavior of the interconnection of switched subsystems and that of their symbolic models.
Moreover, we provide an approach to construct symbolic models together with their corresponding augmented-storage functions for discrete-time switched subsystems under some assumptions ensuring incremental passivity of each mode of switched subsystems. 
Finally, we apply our results to a model of road traffic by constructing compositionally a symbolic model of a network containing $50$ cells of $1000$ meters each. We also design controllers compositionally maintaining the density of traffic lower than $30$ vehicles per cell. Additionally, we apply our results to an interconnection of switched subsystems admitting multiple incrementally passive storage functions.

The results presented in this paper are mainly concerned with the compositional construction of symbolic models of
	interconnected discrete-time switched systems. The constructed symbolic models here can be used to synthesize controllers monolithically or also compositionally. Compositional approaches for controller synthesis can be found in \cite{8115304,meyer} and references therein.

\section{Notation and Preliminaries}\label{1:II}
\subsection{Notation}
We denote by $\R$, $\Z$, and $\N$ the set of real numbers, integers, and non-negative integers, respectively.
These symbols are annotated with subscripts to restrict them in
the obvious way, e.g., $\R_{>0}$ denotes the positive real numbers.
Given $N\in\N_{\ge1}$, vectors $\nu_i\in\R^{n_i}$, $n_i\in\N_{\ge1}$, and $i\in[1;N]$, we
use $\nu=[\nu_1;\ldots;\nu_N]$ to denote the vector in $\R^n$ with
$n=\sum_i n_i$ consisting of the concatenation of vectors~$\nu_i$.
The closed interval in $\N$ is denoted by $[a;b]$ for $a,b\in\N$ and $a\le b$. 
We denote by $\mathsf{diag}(A_1,\ldots,A_N)$ the block diagonal matrix with diagonal matrix entries $A_1,\ldots,A_N$.
We denote the identity matrix in $\R^{n\times n}$ by $I_n$. 
The individual elements in a matrix $A\in \R^{m\times n}$, are denoted by $\{A\}_{ij}$, where $i\in[1;m]$ and $j\in[1;n]$. We denote by $\norm{\cdot}$ the infinity norm. 
We denote by $|\cdot|$ the cardinality of a given set and by $\emptyset$ the empty set. For any set \mbox{$S\subseteq\R^n$} of the form of finite union of boxes, e.g., $S=\bigcup_{j=1}^MS_j$ for some $M\in\N$, where $S_j=\prod_{i=1}^n [c_i^j,d_i^j]\subseteq \R^n$ with $c^j_i<d^j_i$, and positive constant $\eta\leq\emph{span}(S)$, where $\emph{span}(S)=\min_{j=1,\ldots,M}\eta_{S_j}$ and \mbox{$\eta_{S_j}=\min\{|d_1^j-c_1^j|,\ldots,|d_n^j-c_n^j|\}$}, we define \mbox{$[S]_{\eta}=\{a\in S\,\,|\,\,a_{i}=k_{i}\eta,k_{i}\in\mathbb{Z},i=1,\ldots,n\}$}.
The set $[S]_{\eta}$ will be used as a finite approximation of the set $S$ with precision $\eta$. Note that $[S]_{\eta}\neq\emptyset$ for any $\eta\leq\emph{span}(S)$.   
We use notations $\mathcal{K}$ and $\mathcal{K}_\infty$
to denote different classes of comparison functions, as follows:
$\mathcal{K}=\{\alpha:\mathbb{R}_{\geq 0} \rightarrow \mathbb{R}_{\geq 0} |$ $ \alpha$ is continuous, strictly increasing, and $\alpha(0)=0\}$; $\mathcal{K}_\infty=\{\alpha \in \mathcal{K} |$ $ \lim\limits_{r \rightarrow \infty} \alpha(r)=\infty\}$.
\subsection{Discrete-Time Switched and Transition Systems} 
In this work we consider discrete-time switched systems of the following form.
\begin{definition}\label{def:sys1}
	A discrete-time switched system $\Sigma$ is defined by the tuple $\Sigma=(\mathbb X,P,\mathbb W,F,\mathbb Y_1, \mathbb Y_2, h_1, h_2)$,
	where 
	\begin{itemize}
		\item 	$\mathbb X, \mathbb W, \mathbb Y_1$, and $\mathbb Y_2$ are the state set, internal input set, external output set, and internal output set, respectively, and are assumed to be subsets of normed vector spaces with appropriate finite dimensions; 
		\item $P=\{1\cdots,m\}$ is the finite set of modes;
		\item $F=\{f_1,\cdots,f_m\}$ is a collection of set-valued maps $f_p: \mathbb X\times \mathbb W\rightrightarrows\mathbb X $ for all $p\in P$;
		\item $h_1: \mathbb X \rightarrow \mathbb Y_1 $ is the external output map.
		\item $h_2: \mathbb X \rightarrow \mathbb Y_2 $ is the internal output map.
	\end{itemize} 
	The discrete-time switched system $\Sigma $ is described by difference inclusions of the form
		\begin{align}\label{eq:2}
		\Sigma:\left\{
		\begin{array}{rl}
		\mathbf{x}(k+1)&\!\!\!\!\in f_{\mathsf{p}(k)}(\mathbf{x}(k),\omega(k)),\\
		\mathbf{y}_1(k)&\!\!\!\!=h_1(\mathbf{x}(k)),\\
		\mathbf{y}_2(k)&\!\!\!\!=h_2(\mathbf{x}(k)),
		\end{array}
		\right.
		\end{align}where $\mathbf{x}:\mathbb{N}\rightarrow \mathbb X $, $\mathbf{y}_1:\mathbb{N}\rightarrow \mathbb Y_1$, $\mathbf{y}_2:\mathbb{N}\rightarrow \mathbb Y_2$, $\mathsf{p}:\mathbb{N}\rightarrow P$, and $\omega:\mathbb{N}\rightarrow \mathbb W$ are the state signal, external output signal, internal output signal, switching signal, and internal input signal, respectively.
	We denote by $\Sigma_{p}$ the system in \eqref{eq:2} with constant switching signal $\mathsf{p}(k)=p\in P ~\forall k\in \N$.
	We use $\mathbf{X}_{x_0,\overline{p},\overline{\omega}}$ and $\mathbf{Y}_{x_0,\overline{p},\overline{\omega}}$ to denote the sets of infinite state and external output runs of $\Sigma$, respectively, associated with infinite switching sequence $\overline{p}=\{p_0,p_1,\ldots\}$, infinite internal input sequence $\overline{\omega}=\{w_0,w_1,\ldots\}$, and initial state $x_0\in \mathbb{X}$.
	
	Let $\phi_k, k \in \N_{\ge1}, $ denote the time when the $k$-th switching instant occurs and
	define $\Phi := \{\phi_k : k \in \N_{\ge1}\}$ as the set of switching instants.
	We assume that signal $\mathsf{p}$ satisfies a dwell-time condition \cite{Morse} (i.e. there exists $k_d \in \N_{\ge1}$, called the
	dwell-time, such that for all consecutive switching time instants $\phi_k,\phi_{k+1} \in \Phi$, $\phi_{k+1}-\phi_{k}\geq k_d$, for any $k\in\N$). 
	
	System $\Sigma$ is called deterministic if $|f_p(x,w)|\leq1$ $ \forall x\in \mathbb X, \forall p\in P, \forall w \in \mathbb W$, and non-deterministic otherwise. System $\Sigma$ is called blocking if $\exists x\in \mathbb X, \forall p\in P, \forall w \in \mathbb W $ such that $|f_p(x,w)|=0$ and non-blocking if $|f_p(x,w)|\neq 0$ $ \forall x\in \mathbb X, \exists p\in  P, \exists w \in \mathbb W$.  System $\Sigma$ is called finite if $\mathbb X$ and $\mathbb W$ are finite sets and infinite otherwise. In this paper, we only deal with non-blocking systems.
\end{definition}
Next, we introduce a notion of so-called transition systems, inspired by the one in \cite{Girard}, to provide an alternative description of switched systems that can be later directly related to their symbolic models
\begin{definition}\label{tsm} Given a discrete-time switched system $\Sigma\!=\!(\mathbb X,P,\mathbb W,F,\mathbb Y_1, \mathbb Y_2, h_1, h_2)$, we define the associated transition system $T(\Sigma)\!=\!(X,U,W,\TF,Y_1, Y_2, \op_1, \op_2)$. 
	where:
	\begin{itemize}
		\item  $X=\mathbb X\times P\times \{0,\cdots,k_d-1\}$ is the state set; 
		\item $U=P$ is the external input set;
		\item $W=\mathbb{W}$ is the internal input set;
		\item $\TF$ is the transition function given by $(x',p',l')\in \TF((x,p,l),u,w)$ if and only if  $x'\in f_p(x,w),u=p$ and the following scenarios hold:
		\begin{itemize} 
			\item$l<k_d-1$, $p'=p$ and $l'=l+1$: switching is not allowed because the time elapsed since
			the latest switch is strictly smaller than the dwell time;
			\item $l=k_d-1$, $p'=p$ and $l'=k_d-1$: switching is allowed but no switch occurs;
			\item $l=k_d-1$, $p'\neq p$ and $l'=0$: switching
			is allowed and a switch occurs;
		\end{itemize}
		\item $Y_1=\mathbb{Y}_1$ is the external output set;
		\item $Y_2=\mathbb{Y}_2$ is the internal output set;
		\item $\mathcal{H}_1:X\rightarrow Y_1$ is the external output map defined as $\mathcal{H}_1(x,p,l)=h_1(x)$.
		\item $\mathcal{H}_2:X\rightarrow Y_2$ is the internal output map defined as $\mathcal{H}_2(x,p,l)=h_2(x)$.
	\end{itemize}
\end{definition}
We use $T(\mathbf{X})_{z_0,\overline{u},\overline{\omega}}$ and $T(\mathbf{Y})_{z_0,\overline{u},\overline{\omega}}$ to denote the sets of infinite state and external output runs of $T(\Sigma)$, respectively,  associated with infinite external input sequence $\overline{u}=\{u_0,u_1,\ldots\}$, infinite internal input sequence $\overline{\omega}=\{w_0,w_1,\ldots\}$, and initial state $z_0=(x_0,p_0,l_0)\in X$, where $u_0=p_0$ and $l_0=0$.

In the next proposition, we show that sets $\mathbf{Y}_{x_0,\overline{p},\overline{\omega}}$ and $T(\mathbf{Y})_{z_0,\overline{u},\overline{\omega}}$, where $\overline{p}=\overline{u}$ and $z_0=(x_0,p_0,0)$, are equivalent.
\begin{proposition}\label{traj}
	Consider $\Sigma$, $T(\Sigma)$, $\overline{p}=\{p_0,p_1,\ldots\}=\overline{u}$, $\overline{\omega}=\{w_0,w_1,\ldots\}$, and $x_0\in \mathbb{X}$. Then, $\mathbf{Y}_{x_0,\overline{p},\overline{\omega}}=T(\mathbf{Y})_{z_0,\overline{u},\overline{\omega}}$, 
	where $z_0=(x_0,p_0,0)$.
\end{proposition}
\begin{proof} 
	The proof consists of showing that for any infinite run in $\mathbf{Y}_{x_0,\overline{p},\overline{\omega}}$, denoted $\mathbf{y}_{x_0,\overline{p},\overline{\omega}}$, there exists an infinite run in $T(\mathbf{Y})_{z_0,\overline{u},\overline{\omega}}$, denoted $T(\mathbf{y})_{z_0,\overline{u},\overline{\omega}}$, and vise versa. 
	Since $\overline{p}$, $\overline{\omega}$, and $x_0$ are given, one can construct $\mathbf{y}_{x_0,\overline{p},\overline{\omega}}\in \mathbf{Y}_{x_0,\overline{p},\overline{\omega}}$. Then by utilizing the definition of $\mathcal{H}_1$ in Definition \ref{tsm}, we have $\mathbf{y}_{x_0,\overline{p},\overline{\omega}}=T(\mathbf{y})_{z_0,\overline{u},\overline{\omega}}\in T(\mathbf{Y})_{z_0,\overline{u},\overline{\omega}}$.  Now since $\overline{u}=\overline{p}$ and $z_0=(x_0,p_0,0)$, one can construct $T(\mathbf{y})_{z_0,\overline{u},\overline{\omega}}\in T(\mathbf{Y})_{z_0,\overline{u},\overline{\omega}}$. Then again by using the definition of $\mathcal{H}_1$ in Definition \ref{tsm}, we have $T(\mathbf{y})_{z_0,\overline{u},\overline{\omega}}=\mathbf{y}_{x_0,\overline{p},\overline{\omega}}\in \mathbf{Y}_{x_0,\overline{p},\overline{\omega}}$.
\end{proof}
From now on, we use $\Sigma$ and $T(\Sigma)$ interchangeably.

If $\Sigma$ does not have internal inputs, which is the case for interconnected systems (cf. Definition \ref{def:5}), Definition \ref{def:sys1} reduces to the tuple $\Sigma=(\mathbb X,P,F,\mathbb Y,H)$, the set-valued map $f_p$ becomes $f_p:\mathbb X\rightrightarrows\mathbb X$, and \eqref{eq:2} reduces to:
	\begin{align}\label{eq:3}
	\Sigma:\left\{
	\begin{array}{rl}
	\mathbf{x}(k+1)\!\!\!\!&\in f_{\mathsf{p}(k)}(\mathbf{x}(k)),\\
	\mathbf{y}(k)\!\!\!\!&=h(\mathbf{x}(k)).
	\end{array}
	\right.
	\end{align}
Correspondingly, Definition \ref{tsm} reduces to tuple $T(\Sigma)=(X,U,\TF,Y,\op)$, and the transition function $\TF$ is given by $(x',p',l')\in \TF((x,p,l),u)$ if and only if $ x'\in f_p(x),u=p$ and the following scenarios hold:
\begin{itemize}
	\item $l<k_d-1$, $p'=p$ and $l'=l+1$;
	\item $l=k_d-1$, $p'=p$ and $l'=k_d-1$;
	\item $l=k_d-1$, $p'\neq p$ and $l'=0$.
\end{itemize}
\section{Augmented-Storage and Alternating Simulation Functions}\label{I}
Inspired by the definition of the storage function in \cite{7857702}, we introduce a notion of so-called augmented-storage function, which relates two transition systems with internal inputs and outputs.
\begin{definition}\label{sf} 
	Consider $T(\Sigma)\!\!=\!\!\!(X,\!U,\!W,\!\TF,\!Y_1, \!Y_2, \!\op_1,\! \op_2)$  and $\hat{T}(\hat{\Sigma})\!\!=\!\!(\hat{X},\!\hat{U},\!\hat{W},\!\hat{\TF},\!\hat{Y}_1, \!\hat{Y}_2, \!\hat{\op}_1, \hat{\op}_2)$ where $\hat W\subseteq W$ and $\hat Y_1\subseteq Y_1$. A function $ \mathcal{S}:X\times \hat X \to \mathbb{R}_{\geq0} $ is called an augmented-storage function from $\hat{T}(\hat{\Sigma})$ to $T(\Sigma)$ if $\forall (x,p,l)\in X$ and $\forall (\hat{x},p,l)\in\hat{X}$,  one has
		\begin{align}\label{sf1}
		\alpha (\Vert \op_1(x,p,l)\!-\hat{\op}_1(\hat{x},p,l)\Vert ) \!\leq\! \SF(\!(x,p,l),(\hat{x},p,l)\!),
		\end{align}and $\forall (x,p,l)\in X$ and $\forall (\hat{x},p,l)\in\hat{X}$, $\forall \hat u\in\hat{U}$, $\forall w\in W$, $\forall \hat w\in\hat{W}$, $\forall (x',p',l')\in\TF((x,p,l),\hat{u},w)$ $\exists (\hat{x}',p',l')\in\hat{\TF}((\hat{x},p,l),\hat{u},\hat{w})$ such that one gets
		\begin{align}\label{sf2}
		&\SF((x',p',l'),(\hat{x}',p',l'))
		\!\leq \!\sigma \SF((x,p,l),(\hat{x},p,l))\!+\!\varepsilon\!+\!\begin{bmatrix}
		w\!\!-\!\!\hat w\\
		\!\op_2(x,p,l)\!\!-\!\!\hat\op_2(\hat{x},p,l)\!
		\end{bmatrix}^{\!\!T}\!\overbrace{\begin{bmatrix}\!
			R^{11}&\!\!\!R^{12}\!\\
			\!R^{21}&\!\!\!R^{22}\!\!
			\end{bmatrix}}^{R:=}\!\!
		\begin{bmatrix}
		w\!\!-\!\!\hat w\\
		\!\op_2(x,p,l)\!\!-\!\!\hat \op_2(\hat{x},p,l)\!
		\end{bmatrix},
		\end{align}for some $\alpha \in \mathcal{K}_{\infty}$,  $0<\sigma<1$, $\varepsilon\in \mathbb{R}_{\geq 0}$, and some symmetric matrix $R$ of appropriate dimension with conformal block partitions $R^{ij}$. $i,j\in[1;2]$. We say that $\hat{T}(\hat{\Sigma})$ is an abstraction of $T(\Sigma)$ if there exists an augmented-storage function from $\hat{T}(\hat{\Sigma})$ to $T(\Sigma)$. In addition, if $\hat{T}(\hat{\Sigma})$ is finite ($X$ and $W$ are finite sets), we say that $\hat{T}(\hat{\Sigma})$ is a symbolic model of $T(\Sigma)$.
\end{definition}

Now, we introduce a notion of so-called alternating simulation functions, inspired by Definition 1 in \cite{Girard2009566}, which quantitatively relates transition systems without internal inputs and outputs.
\begin{definition}\label{sfg}
	Consider $T(\Sigma)=(X,U,\TF,Y,\op)$  and $\hat{T}(\hat{\Sigma})=(\hat{X},\hat{U},\hat{\TF},\hat{Y},\hat{\op})$ where $\hat Y\subseteq Y$. A function $ \tilde{\SF}:X\times \hat X \to \mathbb{R}_{\geq0} $ is called an alternating simulation function from $\hat{T}(\hat{\Sigma})$ to $T(\Sigma)$ if $\forall (x,p,l)\!\!\in\!\!X$ and $\forall (\hat{x},p,l)\!\!\in\!\!\hat{X}$, one has
		\begin{align}\label{sfg1}
		\tilde{\alpha} (\Vert \op(x,p,l)-\hat{\op}(\hat{x},p,l)\Vert ) \!\leq\! \tilde{\SF}((x,p,l),(\hat{x},p,l)),
		\end{align}and $\forall (x,p,l)\in X$, $\forall (\hat{x},p,l)\in\hat{X}$, $\forall \hat u\in\hat{U}$, $\forall (x',p',l')\in\TF((x,p,l),\hat{u})$ $\exists (\hat{x}',p',l')\in\hat{\TF}((\hat{x},p,l),\hat{u})$   
	such that one gets
		\begin{align}\label{sfg2}
		\tilde{\SF}(\!(x',p',l'),(\hat{x}',p',l')\!)
		\!\leq\! \tilde{\sigma} \tilde{\SF}(\!(x,p,l),(\hat{x},p,l)\!)+\tilde{\varepsilon},
		\end{align}for some $\tilde{\alpha}\in \mathcal{K}_{\infty}$,  $0<\tilde{\sigma}<1$, and $\tilde{\varepsilon}\in \mathbb{R}_{\geq 0}$.
\end{definition}

Note that the notions of storage and simulation
	functions in \cite[Definitions~3.1, 3.2]{7857702} are defined between two continuous-time control systems with continuous state sets, whereas we define the augmented-storage and alternating simulation functions between two transition systems associated with two discrete-time switched systems. Moreover, on the
	right-hand side of  \eqref{sf2} and \eqref{sfg2}, we introduce constant $\varepsilon\in \mathbb{R}_{\geq 0}$ to allow the relation to be defined between two systems with either infinite or finite state sets. The role of $\varepsilon$ will become clear in Section $V$ where we introduce
	symbolic models. Such a constant does not appear in \cite[Definitions~3.1, 3.2]{7857702} which makes them only suitable for systems with continuous state sets.

The next result shows that the existence of an alternating simulation function for transition systems implies the existence of an approximate alternating simulation relation between them as defined in \cite{Tabu}.	
\begin{proposition}\label{error}
	Consider $T(\Sigma)=(X,U,\TF,Y,\op)$ and $\hat{T}(\hat{\Sigma})=(\hat{X},\hat{U},\hat{\TF},\hat{Y},\hat{\op})$ where $\hat Y\subseteq Y$. Assume $ \tilde{\SF}$ is an alternating simulation function from $\hat{T}(\hat{\Sigma})$ to ${T}(\Sigma)$ as in Definition \ref{sfg}. Then, relation $R\subseteq X\times \hat{X}$ defined by $R\!=\!\left\{\!((x,p,l),\!(\hat{x},p,l))\!\in\! {X}\!\times\! \hat{X}|\tilde{\SF}((x,p,l),(\hat{x},p,l))\!\leq \!{\varphi}\!\right\},$ where ${\varphi}=\frac{\tilde{\varepsilon}}{(1-\tilde{\sigma})\psi}$, and $\psi$ can be chosen arbitrarily such that $0<\psi<1$, is an $\hat{\varepsilon}$-approximate alternating simulation relation, defined in \cite{Tabu}, from $\hat{T}(\hat{\Sigma})$ to ${T}(\Sigma)$ with $\hat{\varepsilon}\!\!=\!\tilde{\alpha}^{-1}\!({\varphi})$.
\end{proposition}
\begin{proof}
	The proof consists of showing that $(i)$ $\forall ((x,p,l),(\hat{x},p,l))\in R$ we have $\Vert \op(x,p,l)-\hat{\op}(\hat{x},p,l)\Vert \leq \hat{\varepsilon}$; $(ii)$ $\forall ((x,p,l),(\hat{x},p,l))\in R $ and  $\forall \hat u\in\hat{U}$, $\forall (x',p',l')\in\TF((x,p,l),\hat{u})$ $\exists~ (\hat{x}',p',l')\in\hat{\TF}((\hat{x},p,l),\hat{u})$ satisfying $((x',p',l'),(\hat{x}',p',l'))\in R$. \\
	First observe that \eqref{sfg2} can be written as 
		\begin{align}\label{max}
		\tilde{\SF}(\!(x',p',l'),(\hat{x}',p',l')\!)
		\!\!\leq\!\max\{{\varrho}\tilde{\mathcal{S}}(\!(x,p,l),(\hat{x},p,l)\!),\!{\varphi}\}\!,
		\end{align}
	where $\varrho=1-(1-\psi)(1-\tilde{\sigma})<1$. 
	Then the first item is a simple consequence of the definition of $R$ and condition \eqref{sfg1} (i.e. $\tilde{\alpha} (\Vert \op(x,p,l)-\hat{\op}(\hat{x},p,l)\Vert) \leq \tilde{\SF}((x,p,l),(\hat{x},p,l))\leq\varphi$), which results in $\Vert \op(x,p,l)-\hat{\op}(\hat{x},p,l)\Vert \leq \tilde{\alpha}^{-1}(\varphi)=\hat\varepsilon$. The second item follows immediately from the definition of $R$, inequality \eqref{max}, and the fact that $0<\varrho<1$. In particular, we have $\tilde{\SF}((x',p',l'),(\hat{x}',p',l'))\leq\varphi$ which implies $((x',p',l'),(\hat{x}',p',l'))\in R$.
\end{proof}
\section{Compositionality Result}\label{1:III}
\label{s:inter}
In this section, we consider networks of discrete-time switched subsystems and leverage dissipativity type conditions under which one can construct an alternating simulation function from a network of abstractions to the concrete network by using augmented-storage functions of the subsystems. In the following, we define first a network of discrete-time switched subsystems.
\subsection{Interconnected Systems}
Here, we define the \emph{interconnected discrete-time switched system} as the following.
\begin{definition}\label{def:5}
	Consider $N\in\N_{\ge1}$ switched subsystems $\Sigma_i=(\mathbb X_i,P_i,\mathbb W_i,F_i,\mathbb Y_{1_i}, \mathbb Y_{2_i}, h_{1_i}, h_{2_i})$, and a static matrix $M$ of an appropriate dimension defining the coupling of these subsystems, where\footnote{This condition is required to have a well-defined interconnection.} $M\prod_{i=1}^N \mathbb Y_{2i} \subseteq \prod_{i=1}^N \mathbb W_{i}$. The \emph{interconnected switched
		system} $\Sigma=(\mathbb X,P,F,\mathbb Y,h)$,
	denoted by
	$\mathcal{I}(\Sigma_1,\ldots,\Sigma_N)$, is defined by $\mathbb X \!=\!\prod_{i\!=\!1}^N \!\mathbb X_i$,
	$  P\!=\!\prod_{i\!=\!1}^N  \!P_i$, $F\!=\!\prod_{i\!=\!1}^N  \!F_{i}$, $ \mathbb Y\!=\!\prod_{i\!=\!1}^N  \!\mathbb Y_{1i}$, 
	$h(x)\!\Let \!\intcc{h_{11}(x_1);\ldots;h_{1N}(x_N)}$, where $x=\intcc{x_{1};\ldots;x_{N}}$, with the internal inputs constrained according to
	$\intcc{w_{1};\ldots;w_{N}}=M\intcc{h_{21}(x_1);\ldots;h_{2N}(x_N)}.$
	
	Similarly, given transition subsystem $T_i(\Sigma_i),i\in[1;N]$, one can also define the network of those transition subsystems as $\mathcal{I}(T_1(\Sigma_1),\ldots,T_N(\Sigma_N))$. 
\end{definition}

Next subsection provides one of the main results of the paper on the compositional construction of abstractions for networks of switched systems. 
\subsection{Compositional Abstractions of
	Interconnected Switched Systems}
In this subsection, we assume that we are given $N$ discrete-time switched subsystems $\Sigma_i$, or equivalently, $T_i(\Sigma_i),$ together with their corresponding  abstractions
$\hat{T}_i(\hat{\Sigma}_i)$ and augmented-storage functions $\SF_i$ from
$\hat{T}_i(\hat{\Sigma}_i)$ to $T_i(\Sigma_i)$. 

The next theorem provides a compositional approach on the construction of abstractions of networks of discrete-time switched subsystems and that of the corresponding augmented-storage functions. 
\begin{theorem}\label{thm:3}
	Consider the interconnected transition system
	$T(\Sigma)=(X,U,\TF,Y,\op)$ induced by
	$N\in\N_{\ge1}$
	transition subsystems~$T_i(\Sigma_i),\forall~ i\in [1;N]$. Assume that each $T_i(\Sigma_i)$ and its abstraction $\hat{T}_i(\hat{\Sigma}_i)$ admit an augmented-storage function $\SF_i$ as in Definition \ref{sf}.
	If there exist $\mu_{i}>0$, $i\in[1;N]$, such that the matrix inequality and inclusion
		\begin{align}\label{e:MC1}
		\begin{bmatrix}
		M\\
		I_q 
		\end{bmatrix}^T\overbrace{\begin{bmatrix}
			\tilde{R}_{11}& 	\tilde{R}_{12}\\
			\tilde{R}_{21}& 	\tilde{R}_{22}
			\end{bmatrix}}^{R_{\delta}}\begin{bmatrix}
		M\\
		I_q
		\end{bmatrix}&\preceq0,\\\label{e:MC3}
		M\prod_{i=1}^N {\hat Y}_{2i} \subseteq \prod_{i=1}^N {\hat W}_{i},&
		\end{align}are satisfied, where
	$\tilde{R}_{i'j'}~=~\mathsf{diag}~(\mu_1R_1^{i'j'},\ldots,\mu_NR_N^{i'j'})$, $\forall i',j'\in[1;2]$,
	and $q$ is the number of columns in $M$,
	then 
	\begin{align*}
		\tilde{\SF}((x,p,l),(\hat{x},p,l))\Let\sum_{i=1}^N\mu_i\mathcal{S}_i(({x}_i,p_i,l_i),(\hat{x}_i,p_i,l_i)),
		\end{align*}is an alternating simulation function from $\hat{T}(\hat{\Sigma})={\mathcal{I}}(\hat{T}_1(\hat{\Sigma}_1),\ldots,\hat{T}_{N}(\hat{\Sigma}_{N}))$, with the coupling matrix $M$, to $T(\Sigma)=\mathcal{I}(T_1(\Sigma_1),\ldots,T_{N}(\Sigma_{N}))$.
\end{theorem}
\begin{proof}
	First, we define ${z}=\intcc{ z_1;\ldots;z_N}$, $\hat{z}=\intcc{\hat z_1;\ldots;\hat z_N}$, ${z'}=\intcc{ z'_1;\ldots;z'_N}$, and $\hat{z}'=\intcc{\hat z'_1;\ldots;\hat z'_N}$, where $z_i=(x_i,p_i,l_i)$, $\hat{z}_i=(\hat{x}_i,p_i,l_i)$ $z'_i=(x'_i,p'_i,l'_i)$, and $\hat{z}'_i=(\hat{x}'_i,p'_i,l'_i), ~\forall i\in [1;N]$.
	
	Now, we show that $\eqref{sfg1}$ holds for some $\mathcal{K}_{\infty}$ function $\tilde{\alpha}$. Consider any $z_i\in X_i$, $\hat{z}_i\in {\hat{X}}_i$, $\forall i\in[1;N]$. Then, one gets
		\begin{align*}\notag
		&\Vert \op(z)\!-\!\hat{\op}(\hat{z})\Vert
		\le\!\sum_{i=1}^N \!\Vert \op_{1i}(z_i)\!-\!\hat \op_{1i}(\hat{z}_i)\Vert 
	\le \!\sum_{i=1}^N \!\alpha_{i}^{-1}(\mathcal{S}_i( z_i, \hat z_i))\leq \overline{\alpha}\big(\tilde{\SF}( z, \hat z)\big),
		\end{align*}where $\overline\alpha(s)=\max\limits_{\hat{s}\geq 0}\Bigg\{\sum_{i=1}^N\alpha^{-1}_i(s_i)|\mu^T\hat{s}=s\Bigg\}$, $\hat{s}=\intcc{s_1;\ldots;s_N}\in\R^N$ and $\mu=\intcc{\mu_1;\ldots;\mu_N}$. Hence, $\eqref{sfg1}$ is satisfied with  $\tilde{\alpha}= \overline{\alpha}^{-1}$.
	
	Now, we show that $\eqref{sfg2}$ holds.
	Let $\tilde{\sigma}=\max\limits_{i\in [1,N]}\{\sigma_i\}$, $\tilde{\varepsilon}=\sum_{i=1}^N\mu_i \varepsilon_i$, and consider the following chain of inequalities
		\begin{align}\notag
		\hat{\SF}(z',\hat z')&\!\!=\!\!\sum_{i=1}^N \mu_i\mathcal{S}(z'_{_i},\hat{z}'_{i})\\\label{rewrite}
		&\!\!\leq\!\!\sum_{i=1}^N\!\mu_i\bigg(\!\!\sigma_i\mathcal{S}_i(z_i,\hat{z}_i)+\varepsilon_i
		+\begin{bmatrix}
		w_i\!-\!\hat w_i\\
		\op_{2i}(z_i)\!-\!\hat \op_{2i}(\hat z_i)
		\end{bmatrix}^T\overbrace{\begin{bmatrix}
			R_i^{11}&R_i^{12}\\
			R_i^{21}&R_i^{22}
			\end{bmatrix}}^{R_i:=}
		\begin{bmatrix}
		w_i\!-\!\hat w_i\\
		\op_{2i}(z_i)\!-\!\hat \op_{2i}(\hat z_i)
		\end{bmatrix}\bigg).
		\end{align}Using condition \eqref{e:MC1}, and the definition of matrix $R_{\delta}$, the inequality \eqref{rewrite} can be rewritten as
		\begin{align*}
		\tilde{\SF}(z',\hat{z}')
		\!\leq\!\!&\sum_{i=1}^N\mu_i\sigma_i\mathcal{S}_i(z_i,\hat{z}_i)+\sum_{i=1}^N\mu_i \varepsilon_i
		+\begin{bmatrix}
		\begin{bmatrix}
		w_1\\
		\vdots\\
		w_N
		\end{bmatrix}\!\!-\!\!\begin{bmatrix}
		\hat w_1\\
		\vdots\\
		\hat w_N
		\end{bmatrix}\\
		\op_{21}(z_1)\!\!-\!\!\hat \op_{21}(\hat z_1)\\
		\vdots\\
		\op_{2N}(z_N)-\hat \op_{2N}(\hat z_N)
		\end{bmatrix}^T\!\!\!\!R_{\delta}\!
		\begin{bmatrix}
		\begin{bmatrix}
		w_1\\
		\vdots\\
		w_N
		\end{bmatrix}\!\!-\!\!\begin{bmatrix}
		\hat w_1\\
		\vdots\\
		\hat w_N
		\end{bmatrix}\\
		\op_{21}(z_1)\!\!-\!\!\hat \op_{21}(\hat z_1)\\
		\vdots\\
		\op_{2N}(z_N)\!\!-\!\!\hat \op_{2N}(\hat z_N)
		\end{bmatrix}\\
		\leq&\!\sum_{i=1}^N\!\mu_i\sigma_i\mathcal{S}_i(x_i,\hat{x}_i)\!+\!\!\sum_{i=1}^N\!\mu_i \varepsilon_i\!
		+\!\!\begin{bmatrix}
		\op_{21}(z_1)\!\!-\!\!\hat \op_{21}(\hat z_1)\\
		\vdots\\
		\op_{2N}(z_N)\!\!-\!\!\hat \op_{2N}(\hat z_N)
		\end{bmatrix}^T
		\begin{bmatrix}
		M\\
		I
		\end{bmatrix}^T\!\!\!\!R_{\delta}\begin{bmatrix}
		M\\
		I
		\end{bmatrix}
		\begin{bmatrix}
		\op_{21}(z_1)\!-\!\hat \op_{21}(\hat z_1)\\
		\vdots\\
		\op_{2N}(z_N)\!-\!\hat \op_{2N}(\hat z_N)
		\end{bmatrix}\\
		\leq&\sum_{i=1}^N\mu_i\sigma_i\mathcal{S}_i(z_i,\hat{z}_i)+\sum_{i=1}^N\mu_i \epsilon_i\\
		\leq&\tilde{\sigma}\tilde{\SF}(z,\hat z)+\tilde{\varepsilon}.
		\end{align*}which satisfies $\eqref{sfg2}$, and implies that $\tilde{\SF}$ is indeed an alternating simulation function from $\hat{T}(\hat{\Sigma})$ to $T(\Sigma)$. 
\end{proof}
\begin{remark}\label{setw}
	Condition \eqref{e:MC1} is a linear matrix inequality which can be verified by some semi-definite programming tools (e.g. YALMIP \cite{YALMIP}). 
	Note that condition \eqref{e:MC3} is required to have a well-defined interconnection of abstractions and is automatically fulfilled if one constructs the internal input sets of each abstractions $\hat{T}_i(\hat{\Sigma}_i)$ such that the equality $M\prod_{i=1}^N {\hat Y}_{2i} = \prod_{i=1}^N {\hat W}_{i}$ holds.
\end{remark}

Remark that similar compositionality result as in Theorem \ref{thm:3} was proposed in \cite{7857702}. Since \cite{7857702} is concerned with \emph{infinite} abstractions (a continuous-time control system with potentially a lower dimension), extra matrices (i.e. $W$, $\hat{W}$, $H$ in \cite[equation (9)]{7857702}) are required to formulate the dissipativity-type conditions. However, as our work is mainly concerned with symbolic models, we formulate the dissipativity-type conditions without requiring those extra matrices.
\section{Construction of Symbolic Models}\label{1:IV}\label{SM}
In this section, we consider $\Sigma\!\!=\!\!(\!\mathbb X,\!P,\!\mathbb W,\!F,\!\mathbb Y_1, \!\mathbb Y_2, \!h_1,\! h_2\!)$ as an infinite, deterministic switched system, 
and assume its external output map $h_1$ satisfies the following general Lipschitz assumption: there exists $\ell\!\in\!\KK$ such that: $\Vert h_1(x)\!-\!h_1(x')\Vert\leq \ell(\Vert x\!-\!x'\Vert)$ $\forall x,x'\in \mathbb X$. 
In addition, the existence of an augmented-storage function between $T(\Sigma)$ and its symbolic model is established under the assumption that $\Sigma_p$ is so-called incrementally passive ($\delta$-P) \cite{arxiv} as defined next.
\begin{definition}\label{def:SFD1} 
	System $\Sigma_{p}$  is $\delta$-P if there exist functions $ S_p:\mathbb X\times \mathbb X \to \mathbb{R}_{\geq0} $, $\underline{\alpha}_p \in \mathcal{K}_{\infty}$, a symmetric matrix $Q_p$ of appropriate dimension, and constant $0<\kappa_p<1$, such that for all $x,\hat x\in \mathbb{X}$, and for all $w,\hat w\in \mathbb{W}$  
		\begin{align}\label{e:SFC11}
		\underline{\alpha}_p (\Vert x-\hat{x}\Vert ) \leq S_p(x,\hat{x})
		\end{align}
		\begin{align}\label{e:SFC22}
		S_p&(f_p(x,w),f_p(\hat x,\hat w))
		\!\leq\! \kappa_p S_p(x,\hat{x})+\begin{bmatrix}
		w-\hat w\\
		h_2(x)-h_2(\hat{x})
		\end{bmatrix}^T\overbrace{\begin{bmatrix}
			Q^{11}_{p}&Q^{12}_{p}\\
			Q^{21}_{p}&Q^{22}_{p}
			\end{bmatrix}}^{Q_{p}:=}
		\begin{bmatrix}
		w-\hat w\\
		h_2(x)-h_2(\hat x)
		\end{bmatrix}.\\\notag
		\end{align}
\end{definition}

We say that $S_p$ and $Q_p$, $\forall p\in P$, are multiple $\delta$-P storage functions and supply rates, respectively, for system $\Sigma$ if they satisfy \eqref{e:SFC11} and \eqref{e:SFC22}. Moreover, if $S_p=S_{p'}$ and $Q_p=Q_{p'}$, $\forall p,p'\in P$, we omit the index $p$ in \eqref{e:SFC11}, \eqref{e:SFC22}, and say that $S$ and $Q$ are a common $\delta$-P storage function and supply rate for system $\Sigma$.
 
Now, we show how to construct a symbolic model $\hat T(\hat{\Sigma})$ of transition system $T(\Sigma)$ associated to the switched system $\Sigma$ where $\Sigma_{p}$  is $\delta$-P.
\begin{definition}\label{smm} Consider a transition system $T(\Sigma)=(X,U,W,\TF,Y_1, Y_2, \op_1, \op_2)$, associated to the switched system $\Sigma=(\mathbb X,P,\mathbb W,F,\mathbb Y_1, \mathbb Y_2, h_1, h_2)$, where $\mathbb X,\mathbb W$ are assumed to be finite unions of boxes. Let $\Sigma_{p}$ be $\delta$-P as in Definition \ref{def:SFD1}. Then one can construct a finite transition system (a symbolic model) $\hat{T}(\hat{\Sigma})=(\hat{X},\hat{U},\hat{W},\hat{\TF},\hat{Y}_1, \hat{Y}_2, \hat{\op}_1, \hat{\op}_2)$ where:
	\begin{itemize}
		\item $\hat{X}=\hat{\mathbb{X}}\times P\times \{0,\cdots,k_d-1\}$, where $\hat{\mathbb{X}}=[\mathbb{X}]_{\eta}$ and $0<\eta\leq\emph{span}(\mathbb{X})$ is the state set quantization parameter; 
		\item $\hat{U}=U=P$ is the external input set;
		\item $\hat{W}=[\mathbb{W}]_{{\varpi}}$, where $0\leq{{\varpi}}\leq\emph{span}(\mathbb W)$ is the internal input set quantization parameter.
		\item $(\hat{x}',p',l')\in \hat{\TF}((\hat{x},p,l),\hat{u},\hat{w})$ if and only if $\Vert f_p(\hat{x},\hat{w})-\hat{x}'\Vert\leq \eta$, $\hat{u}=p$ and the following scenarios hold:
		\begin{itemize}
			\item $l<k_d-1$, $p'=p$ and $l'=l+1$;
			\item $l=k_d-1$, $p'=p$ and $l'=k_d-1$;
			\item $l=k_d-1$, $p'\neq p$ and $l'=0$;
		\end{itemize}
		\item $\hat{Y}_1=Y_1,\hat{Y}_2=Y_2$;
		\item $\hat{\op}_1:\hat{X}\rightarrow \hat{Y}_1$ is the external output map defined as $\hat{\op}_1(\hat{x},p,l)={\op}_1(\hat{x},p,l)=h_1(\hat{x})$;
		\item $\hat{\op}_2:\hat{X}\rightarrow \hat{Y}_2$ is the internal output map defined as $\hat{\op}_2(\hat{x},p,l)={\op}_2(\hat{x},p,l)=h_2(\hat{x})$;
	\end{itemize} 
\end{definition}

\begin{remark} Although one can freely construct $\hat{W}$, in the context of networks of subsystems, it should be constructed 
	in such a way that the interconnection of finite transition subsystems is well-defined (cf. Remark \ref{setw}).   
\end{remark}

Let us point out some differences between the symbolic model in Definition \ref{smm} and the one proposed in \cite{Girard}. There is no
	distinction between internal and external inputs and outputs in the symbolic model defined in \cite{Girard}, whereas their distinctions in our work play a major role in interconnecting subsystems and providing the main compositionality result.

In the following, we impose assumptions on function $S_p$ in Definition \ref{def:SFD1} which are used to prove some of the main results later. 
\begin{assumption}\label{ass1} 
	There exists $\mu \geq 1$ such that
		\begin{align}\label{mue}
		\forall x,y \in \mathbb{X},~~ \forall p,p' \in P,~~ S_p(x,y)\leq \mu S_{p'}(x,y). 
		\end{align}
\end{assumption}
\vspace{0.1cm}
Assumption \ref{ass1} is an incremental version of a similar assumption that is used to prove input-to-state stability of switched systems under constrained switching assumptions \cite{VU}.
\begin{assumption}\label{ass2} 
	Assume that $\forall p\!\in\! P$, $\exists\gamma_p\!\in\!\mathcal{K}_{\infty}$ such that
		\begin{align}\label{tinq} 
		\forall x,y,z \in \mathbb{X},~~S_p(x,y)\leq S_p(x,z)+\gamma_p(\Vert y-z\Vert).
		\end{align}
\end{assumption}
\vspace{0.1cm}
Assumption \ref{ass2} is shown in \cite{zamani2014symbolic} to be a non-restrictive condition provided that one is interested to work on a compact subset of $\mathbb{X}\times\mathbb{X}$.

Now, we establish the relation between $T(\Sigma)$ and $\hat{T}(\hat{\Sigma})$, introduced above, via the notion of augmented-storage function as in Definition \ref{sf}. 

\begin{theorem}\label{thm1}
	Consider a switched system $\Sigma=(\mathbb X,P,\mathbb W,F,\mathbb Y_1, \mathbb Y_2, h_1, h_2)$ with its equivalent transition system $T(\Sigma)=(X,U,W,\TF,Y_1, Y_2, \op_1, \op_2)$. Let $\Sigma_{p}$ be $\delta$-P as in Definition \ref{def:SFD1}. Consider a finite transition system $\hat{T}(\hat{\Sigma})=(\hat{X},\hat{U},\hat{W},\hat{\TF},\hat{Y}_1, \hat{Y}_2, \hat{\op}_1, \hat{\op}_2)$ constructed as in Definition \ref{smm}. Suppose that Assumptions \ref{ass1} and \ref{ass2} hold. Let $\epsilon>1$ and define ${\kappa}\!=\!\max_{p\in P}\left\{\kappa_p\right\}$. If, $k_d\geq \epsilon\frac{\ln (\mu)}{\ln (\frac{1}{{\kappa}})}+1$, and there exists a symmetric  matrix $\tilde{Q}$ such that $\forall q\in \{1,\ldots,k_d-1\},\tilde{Q}-\kappa^{\frac{-q}{\epsilon}}\sum_{p=1}^{m}Q_p\succeq0$,
	then function $\mathcal{V}$ defined as 
		\begin{align}\label{sm}
		\mathcal{V}((x,p,l),(\hat{x},p,l))\Let{\kappa^{\frac{-l}{\epsilon}}\sum_{p=1}^{m}S_p(x,\hat{x})},
		\end{align}is an augmented-storage function from $\hat{T}(\hat{\Sigma})$ to $T(\Sigma)$.
\end{theorem}
\begin{proof} Given the Lipschitz assumption on $h_1$ and since, $\forall p\in P$, $\Sigma_p$ is $\delta$-P , from \eqref{e:SFC11}, $\forall (x,p,l)\in X$ and $ \forall (\hat{x},p,l) \in \hat{X}
	$, we have 
		\begin{align*}
	&\Vert \op_1(x,p,l)-\hat{\op}_1(\hat{x},p,l)\Vert\!=\!\Vert h_1(x)-\hat{h}_1(\hat{x})\Vert
		\!\leq \!\ell(\Vert x-\hat{x}\Vert)
		\!\leq\!\ell\circ\underline{\alpha}^{-1}_p(S_p(x,\hat{x}))
		\!<\!\ell\circ\underline{\alpha}^{-1}_p\left(\sum_{p=1}^{m}S_p(x,\hat{x})\right)\\
		&\!=\!\ell\circ\underline{\alpha}^{-1}_p\left({\kappa^{\frac{l}{\epsilon}}}\mathcal{V}((x,p,l),(\hat{x},p,l))\right)
		\!\leq\!\ell\circ\underline{\alpha}^{-1}_p\left(\mathcal{V}((x,p,l),(\hat{x},p,l))\right)
	\!\leq\!\hat{\alpha}\left(\mathcal{V}((x,p,l),(\hat{x},p,l))\right),
		\end{align*}
	where $\hat{\alpha}\!=\!\max\limits_{p\in P}\{\ell\circ\underline{\alpha}^{-1}_p\}$. Hence \eqref{sf1} is satisfied with $\alpha\!=\!\hat{\alpha}^{-1}$.

	Now from \eqref{tinq} and Definition \ref{smm}, $\forall x\!\in \!\mathbb{X}, \forall \hat{x} \!\in\! \mathbb{\hat{X}},\forall w \!\in\! \mathbb{W},\forall \hat{w} \in {\hat{\mathbb{W}}}$, we have 
	
		\begin{align*}
		S_p(f_p(x,w),\hat{x}')&\leq S_p(f_p(x,w),f_p(\hat{x},\hat{w}))\!+\!\gamma_p(\Vert \hat{x}'\!\!-\!\!f_p(\hat{x},\hat{w})\Vert)\\
		&\leq S_p(f_p(x,w),f_p(\hat{x},\hat{w}))+\gamma_p(\eta).
		\end{align*}for any $\hat{x}'$ such that  $(\hat{x}',p',l')\in \hat{\TF}((\hat{x},p,l),\hat{u},\hat{w})$.
	Let $\mathcal{T}(w,x,\hat{w},\hat{x})\Let [w-\hat w;h_2(x)\!-\!h_2(\hat{x})]$ and note that by \eqref{e:SFC22}, one gets 
		\begin{align*}
		S_p&(f_p(x,w),f_p(\hat x,\hat w))\leq \!\kappa_p S_p(x,\hat{x})
		+\mathcal{T}(w,x,\hat{w},\hat{x})^TQ_{p}
		\mathcal{T}(w,x,\hat{w},\hat{x}).
		\end{align*}
	Hence, $\forall x\in \mathbb{X}, \forall \hat{x} \in \mathbb{\hat{X}}$, and $\forall w \in \mathbb{W},\forall \hat{w} \in {\hat{\mathbb{W}}}$, one obtains
		\begin{align}\label{fe}
		S_p&(f_p(x,w),\hat{x}')\leq \!\kappa_p S_p(x,\hat{x})
		\!\!+\!\!\mathcal{T}(w,x,\hat{w},\hat{x})^T\!Q_{p}\!
		\mathcal{T}(w,x,\hat{w},\hat{x})\!+\!\gamma_p(\eta),
		\end{align}for any $\hat{x}'$ such that  $(\hat{x}',p',l')\in \hat{\TF}((\hat{x},p,l),\hat{u},\hat{w})$.
	Now, in order to show function $\mathcal{V}$ defined in \eqref{sm} satisfies \eqref{sf2}, we consider the different scenarios in Definition \ref{smm} as follows.  
		\begin{itemize}
			\item$l<k_d-1$, $p'=p$ and $l'=l+1$, using \eqref{fe}, we have
			\begin{small}
				\begin{align*}
				&\mathcal{V}(\!(x',p',l'),(\hat{x}',p',l')\!)\!=\!\frac{\sum_{p'=1}^{m}S_{p'}(x',\hat{x}')}{\kappa^{\frac{l'}{\epsilon}}}\!\!=\!\!\frac{\sum_{p=1}^{m}S_p(f_p(x,w),\hat{x}')}{\kappa^{\frac{l+1}{\epsilon}}}
				\!\!\leq\!\!\frac{\sum_{p=1}^{m}\kappa_pS_p(x,\hat{x})}{\kappa^{\frac{1}{\epsilon}}\kappa^{\frac{l}{\epsilon}}}\!+\!\frac{\sum_{p=1}^{m}(\mathcal{\tilde{T}}(w,x,\hat{w},\hat{x},Q_{p})\!+\!\gamma_p(\eta))}{\kappa^{\frac{l+1}{\epsilon}}}\\
				&\leq{\kappa}^{\frac{\epsilon-1}{\epsilon}}\mathcal{V}((x,p,l),(\hat{x},p,l))\!+\!\frac{\sum_{p=1}^{m}\mathcal{\tilde{T}}(w,x,\hat{w},\hat{x},Q_{p})\!}{\kappa^{\frac{l+1}{\epsilon}}}+\!\frac{\sum_{p=1}^{m}\!\gamma_p(\eta)}{\kappa^{\frac{k_d}{\epsilon}}}.
				\end{align*}
			\end{small}
			\item $l=k_d-1$, $p'=p$ and $l'=k_d-1$, using \eqref{fe} and $\frac{\epsilon-1}{\epsilon}<1$, one gets
			\begin{small}
				\begin{align*}
				&\mathcal{V}(\!(x',p',l'),(\hat{x}',p',l')\!)\!=\!\frac{\sum_{p'=1}^{m}S_{p'}(x',\hat{x}')}{\kappa^{\frac{l'}{\epsilon}}}\!\!=\!\!\frac{\sum_{p=1}^{m}S_p(f_p(x,w),\hat{x}')}{\kappa^{\frac{l}{\epsilon}}}
				\!\!\leq\!\!\frac{\sum_{p=1}^{m}\kappa_pS_p(x,\hat{x})}{\kappa^{\frac{l}{\epsilon}}}\!+\!\frac{\sum_{p=1}^{m}(\mathcal{\tilde{T}}(w,x,\hat{w},\hat{x},Q_{p}) \!+\!\gamma_p(\eta))}{\kappa^{\frac{l}{\epsilon}}}\\
				&\leq{\kappa}^{\frac{\epsilon-1}{\epsilon}}\mathcal{V}((x,p,l),(\hat{x},p,l))\!+\!\frac{\sum_{p=1}^{m}\mathcal{\tilde{T}}(w,x,\hat{w},\hat{x},Q_{p})}{\kappa^{\frac{l}{\epsilon}}}+\!\frac{\sum_{p=1}^{m}\!\gamma_p(\eta)}{\kappa^{\frac{k_d}{\epsilon}}}.
				\end{align*}
			\end{small}
			\item $l=k_d-1$, $p'\neq p$ and $l'=0$, using \eqref{fe}, $k_d\geq \epsilon\frac{\ln (\mu)}{\ln (\frac{1}{{\kappa}})}+1\Leftrightarrow\mu{\kappa^{\frac{k_d-1}{\epsilon}}_{p}}\leq1$, and $\frac{\epsilon-1}{\epsilon}<1$, one has
			\begin{small}
				\begin{align*}
				&\mathcal{V}(\!(x',p',l'),(\hat{x}',p',l')\!)\!=\!\!\frac{\sum_{p'=1}^{m}S_{p'}(x',\hat{x}')}{\kappa^{\frac{l'}{\epsilon}}}\!\!\leq\!\!\mu\!\sum_{p=1}^{m} S_{p}(f_{p}(x,w),\hat{x}')
				\!\!\leq\!\!\frac{\mu\kappa^{\frac{k_d-1}{\epsilon}}\!\!\left(\sum_{p=1}^{m}(\kappa_{p}S_{p}(x,\hat{x})\!\!+\!\!\mathcal{\tilde{T}}(w,x,\hat{w},\hat{x},Q_{p})\!+\!\gamma_{p}(\eta))\right)}{\kappa^{\frac{k_d-1}{\epsilon}}}\\
				&\!\leq\!\frac{\sum_{p=1}^{m}\kappa_pS_p(x,\hat{x})}{\kappa^{\frac{k_d-1}{\epsilon}}}\!\!+\!\!\frac{\sum_{p=1}^{m}\!(\mathcal{\tilde{T}}(w,x,\hat{w},\hat{x},Q_{p})\!+\!\gamma_p(\eta))}{\kappa^{\frac{k_d-1}{\epsilon}}}
				\!\!\leq\!\!{\kappa}^{\frac{\epsilon-1}{\epsilon}}\mathcal{V}((x,p,l),(\hat{x},p,l))\!\!+\!\!\frac{\sum_{p=1}^{m}\mathcal{\tilde{T}}(w,x,\hat{w},\hat{x},Q_{p})}{\kappa^{\frac{k_d-1}{\epsilon}}}\!\!+\!\!\frac{\sum_{p=1}^{m}\!\gamma_p(\eta)}{\kappa^{\frac{k_d}{\epsilon}}}.
				\end{align*}
			\end{small}
		\end{itemize}
	Let $\tilde{\gamma}=\kappa^{\frac{-k_d}{\epsilon}}\sum_{p=1}^{m}\gamma_{p}$, $\forall (x,p,l)\!\in\! X$, $ \forall (\hat{x},p,l) \!\in\! \hat{X}$, $\forall w \!\in\! W$, and $\forall \hat{w} \in {\hat{W}}$. Since $\hat{\op}_2(\hat{x},p,l)=h_2(\hat{x})$ and ${\op}_2({x},p,l)=h_2({x})$, one obtains
		\begin{align*}
		\mathcal{V}&((x',p',l'),(\hat{x}',p',l'))
		\leq{\kappa}^{\frac{\epsilon-1}{\epsilon}}\mathcal{V}((x,p,l),(\hat{x},p,l))+\tilde{\gamma}(\eta)+\!\!\begin{bmatrix}
		w\!-\!\hat w\\
		\op_2(x,p,l)\!-\!\op_2(\hat{x},p,l)
		\end{bmatrix}^T\!\!\tilde{Q}
		\begin{bmatrix}
		w\!-\!\hat w\\
		\op_2(x,p,l)\!-\!\hat \op_2(\hat{x},p,l)
		\end{bmatrix}\!\!,
		\end{align*} 
	Hence, inequality \eqref{sf2} is satisfied 
	with $\sigma={\kappa}^{\frac{\epsilon-1}{\epsilon}}$, $R=\tilde{Q}$, $\varepsilon=\tilde{\gamma}(\eta)$. Thus, $\mathcal{V}$ is an augmented-storage function from  $\hat{T}(\hat{\Sigma})$ to $T(\Sigma)$. 
	Using exactly the same argument, we can show the $\mathcal{V}$ is an augmented-storage function from from $T(\Sigma)$ to $\hat{T}(\hat{\Sigma})$.  		
\end{proof}
\begin{remark}\label{nonl-common}
	If equation \eqref{e:SFC22} is satisfied with the same $Q_p, \forall p\!\!\in \!\!P$, then function $\mathcal{V}$ in Theorem \ref{thm1} reduces to $\mathcal{V}((x,p,l),(\hat{x},p,l))\!\!\Let \!\!{\kappa^{\frac{-l}{\epsilon}}S_p(x,\hat{x})}$. In addition, if $\Sigma$ admits a common $\delta$-P storage function, function $\mathcal{V}$ reduces to $\mathcal{V}((x,p,l),(\hat{x},p,l))\Let S(x,\hat{x})$.
\end{remark}
\begin{remark}\label{linear}
	For affine switched systems$\big(i.e.,\mathbf{x}(k\!+\!1)= A_{\mathsf{p}(k)}\mathbf{x}(k)\!\!+\!\!D_{\mathsf{p}(k)}\omega(k)\!\!+\!\!B_{\mathsf{p}(k)},\mathbf{y}_1(k)\!\!=\!\!C_1\mathbf{x}(k),\mathbf{y}_2(k)\!\!=\!\!C_2\mathbf{x}(k)\big)$, we can restrict attention to $\delta$-P storage functions of the form $S_p(x,\hat{x})\!=\!{(x\!-\!\hat x)^TZ_p(x\!-\!\hat x)}, Z_p\!\!\succ\! 0$. It is readily seen that such functions always satisfy \eqref{e:SFC11} and \eqref{mue}. Moreover, inequality \eqref{e:SFC22} reduces to the linear matrix inequality 
		\begin{align}\label{lmi}
		\!\!\!\!\!\begin{bmatrix}
		\theta_p A_p^TZ_pA_p & A^T_pZ_pD_p \\
		D^T_pZ_pA_p & \theta_p D^T_pZ_pD_p
		\end{bmatrix}\!\!\preceq\!\!\begin{bmatrix}
		\kappa_p Z_p\!+\!C_2^TQ_p^{22}C_2&  C_2^TQ_p^{21}\\
		Q^{12}_pC_2&  Q^{11}_p\end{bmatrix}\!\!
		\end{align}in which $Z_p$ and $Q_p$ can be determined by semi-definite programming, where $\theta_p>1, 0<\kappa_p<1$.
	Consequently, it can be readily verified that $\varepsilon$ in \eqref{sf2} would be defined as $\varepsilon=c_p\lambda_{\max}(\!Z_p\!)$, for some $c_p\!\!>\!0$ depending on $\theta_p$ and the dimensions of $Z_p$.
\end{remark}

\section{Case Study}\label{1:V}

\subsection{Model of road traffic}\label{1}
Consider the switched system $\Sigma$ which is adapted from \cite{Kibangou} and described by  
	\begin{align*}
	\Sigma:\left\{
	\begin{array}{rl}
	\mathbf{x}(k+1)\!&=A\mathbf{x}(k)+ B_{\mathsf{p}(k)},\\
	\mathbf{y}(k)\!&=\mathbf{x}(k),
	\end{array}\right.
	\end{align*} where $A\in\R^{50\times 50}$ is a matrix with elements $\{A\}_{qq}=0.9-\frac{\tau v}{d}$ if $q\in Q_1=\{q ~\text{is} ~\text{odd}~ | q\in [1;50]\}$ and $\{A\}_{qq}=0.65-\frac{\tau v}{d}$ if $q\in Q_2=\{q ~\text{is}~\text{even}~ | q\in [1;50]\}$, $\{A\}_{(q+1)q}=\{A\}_{1(50)}=\frac{\tau v}{d}$, $\forall q\in [1;50]$, and all other elements are identically zero,
where $\tau=\frac{10}{60\times60}$, $d=1$, and $v=120$ are sampling time interval in hours, length in kilometers, and the flow
speed of the vehicles in kilometers per hour, respectively. The vector $B_p\in \R^{50}$ is defined as ${B}_p=\intcc{b_{1p_1};\ldots; b_{25p_{25}}}$ such that $b_{ip_i}=[0;0]$ if $p_i=1$, and 
$b_{ip_i}=[0;12]$ if $p_i=2$, $\forall i\in[1,25]$, $\intcc{p_{1};\ldots;p_{25}}\in P=\{1,2\}^{25}$, where $P$ is the set of modes of $\Sigma$.

The chosen switched system $\Sigma$ here is the
model of a circular road around a city (Highway) divided in $50$ cells of $1000$ meters each. The road has
$25$ entries and $50$ exits. A cell $q$ has an entry and exit if $q\in Q_1$ and has an exit and no entry if $q\in Q_2$. All the entries are controlled by traffic signals, denoted $s_r, r\in[1;25]$.
In $\Sigma$, the dynamic we want to observe is the density of traffic, given in vehicles per cell, for each cell $q$
of the road. 
During the sampling time interval $\tau$, we assume that $12$ vehicles can pass the entry controlled by a traffic signal $s_r$ when it is green. Moreover, $10\%$ of vehicles that are in cells $q\in Q_1$, and $35\%$ of vehicles that are in cells $q\in Q_2$ go out using available exits. 

\begin{figure}
	\begin{center}
		\includegraphics[height=6.5cm]{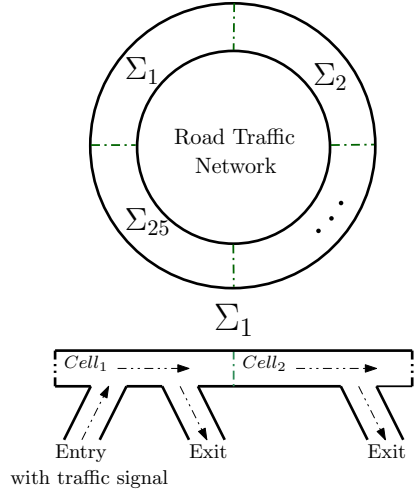}
		\caption{{\small Model of a road traffic network in a circular highway composed of 25 identical links, each link has two cell}.}
		\label{rdn}
	\end{center}
\end{figure}
Now, in order to apply the compositionality result, we introduce subsystems $\Sigma_i$, $\forall i \in [1;25]$. Each subsystems $\Sigma_i$ represents the dynamic of one link of the entire highway, where each link contains $2$ cells, one entry, and two exits, as schematically illustrated in Figure \ref{rdn}. The subsystems $\Sigma_i$ is described by
	\begin{align*}
	\Sigma_i:\left\{
	\begin{array}{rl}
	\mathbf{x}_i(k+1)&=A_i\mathbf{x}_i(k)+ D_{i}w_i(k)+B_{i\mathsf{p}_i(k)},\\
	\mathbf{y}_{1i}(k)&=\mathbf{x}_i(k),\\
	\mathbf{y}_{2i}(k)&=C_{2i}\mathbf{x}_i(k),
	\end{array}\right.\\
	A_i\!\!=\!\!\begin{bmatrix}\!
	0.9\!-\!\!\frac{\tau v}{d} \!\!& 0  \\
	\frac{\tau v}{d} \!\!& 0.65\!-\!\!\frac{\tau v}{d}
	\!\end{bmatrix}\!\!,
	D_i\!\!=\!\!\begin{bmatrix}
	\frac{\tau v}{d} \\
	0\\
	\end{bmatrix}\!\!,
	B_{i1}\!\!=\!\!\begin{bmatrix}
	0  \\
	0\\
	\end{bmatrix}\!\!,\!
	B_{i2}\!\!=\!\!\begin{bmatrix}
	12  \\
	0\\
	\end{bmatrix}\!\!,
	C_{2i}\!\!=\!\!\begin{bmatrix}
	0\\1
	\end{bmatrix}^T\!\!\!\!\!\!,
	\end{align*}
and the set of modes is $P_i=\{1,2\}$, $\forall i\in[1;25]$. 
Clearly, $\Sigma=\mathcal{I}(\Sigma_1,\ldots,\Sigma_{25})$, where the elements of the coupling matrix $M$ are $\{M\}_{(i+1)i}=\{M\}_{1(25)}=1$, $\forall i\in [1;25]$, and all other elements are identically zero.

Note that, for any $i\in[1;25]$,
conditions \eqref{e:SFC11} and \eqref{e:SFC22} are satisfied with $S_{ip_i}(x_i,\hat{x}_i)\!=\!{(x_i-\hat x_i)^TZ_{ip_i}(x_i-\hat x_i)}$, $Z_{ip_i}=I_2$, $\underline{\alpha}_{ip_i}(s)=s^2$, $\kappa_{ip_i}= 0.98$, $Q_{ip}^{11}=0.3527$, $Q_{ip}^{12}=Q_{ip}^{21}=0.0937$, $Q_{ip}^{22}=-0.6785$ $\forall p_i\in P_i$.  
Moreover, since $S_{ip_i}= S_{ip'_i},\forall p,p'\in P$, and according to Remarks \ref{nonl-common} and \ref{linear}, function $\mathcal{V}_i((x_i,p_i,l_i),(\hat{x}_i,p_i,l_i))= \!S_{i}(x_i,\hat{x}_i)$ is an augmented-storage function from $\hat{T}_i(\hat{\Sigma}_i)$, constructed as in Definition \ref{smm}, to $T_i(\Sigma_i)$, defined in Definition \ref{tsm}.
Now, by choosing $\mu_i=1,\forall i\in[1;25]$ and finite internal input sets $\hat{W}_i$ of $\hat{T}_i(\hat\Sigma_i)$ in such a way that $\prod_{i=1}^{25}\!\hat{W}_i\!\!=\!\!M\prod_{i\!=\!1}^{25}\!\hat{X}_i$, condition \eqref{e:MC1} and \eqref{e:MC3} are satisfied.
Therefore, applying Theorem \ref{thm:3}, function $\tilde{\mathcal{S}}((x,p,l),(\hat{x},p,l))\!=\!\!\sum_{i=1}^{25}\!\mathcal{V}_i((x_i,p_i,l_i),(\hat{x}_i,p_i,l_i))$
is an alternating simulation function from $\hat{\mathcal{I}}(\hat{T}_1(\hat{\Sigma}_1),\ldots,\hat{T}_{25}(\hat{\Sigma}_{25}))$ to 
$\mathcal{I}(T_1(\Sigma_1),\ldots,T_{25}(\Sigma_{25}))$.

Let us now design a controller for $\Sigma$ via symbolic models $\hat{T}_i(\hat{\Sigma}_i)$ such that controllers maintain
	the density of traffic lower than $30$ vehicles per cell (safety constraint), and to allow only 2 consecutive red lights for each traffic signal (fairness constraint). The former
	constraint implies that each vehicle can keep a $30$-meter safe distance from the one directly in front. The latter
	constraint is a way to avoid the trivial solution (always red) of the safety constraint and ensures fairness between modes 1 and 2. The idea here is to design local controllers
	for symbolic models $\hat{T}_i(\hat{\Sigma}_i)$, and then refine them to the ones for concrete switched subsystems $\Sigma_i$. To do so, the local controllers are designed
	while assuming that the other subsystems meet their specifications. This approach, called assume-guarantee reasoning \cite{Rajamani}, allows for the compositional synthesis of controllers.

Note that the direct computation of the symbolic model for the original $50$-dimensional system $\Sigma$ is not possible monolithically. To the best of our knowledge, there does not exist any software toolbox for constructing symbolic models of systems with this number of state variables. On the other hand, we are able to construct the interconnected symbolic model and controllers for the $50$-dimensional system $\Sigma$ by applying the proposed compositionality method here. We leverage software tool \texttt{SCOTS} \cite{Rungger} for constructing symbolic models and controllers for $\Sigma_i$ compositionally with the state quantization parameter $\eta_i=0.03$ and the computation times are amounted to $10.2s$ and $0.014s$, respectively. Figure \ref{st} shows the applied modes of sample subsystem $\Sigma_i$. Moreover, the closed-loop state trajectories of $\Sigma$, consisting of $50$ cells, are illustrated in Figure \ref{st1}.

\begin{figure}
	\begin{center}
		\includegraphics[height=5.5cm, width=10cm]{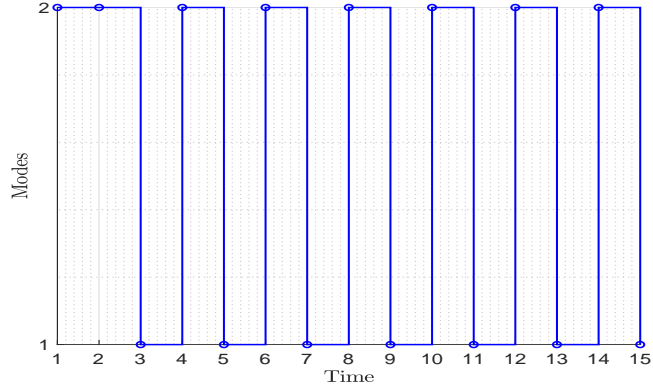}
		\caption{{\small Applied modes of sample subsystem $\Sigma_i$.}}
		\label{st}
	\end{center}
\end{figure}
\begin{figure}
	\begin{center}
		\includegraphics[height=5.5cm, width=10cm]{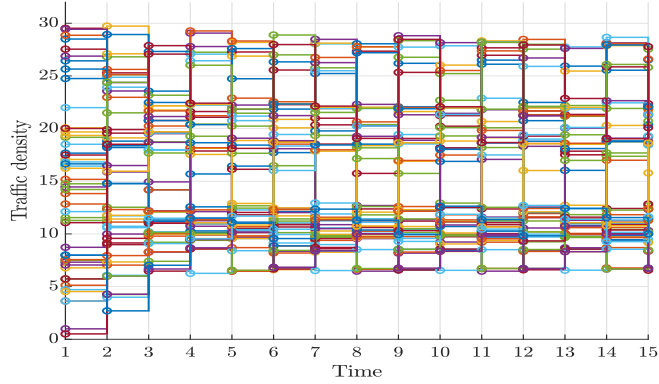}
		\caption{{\small Closed-loop state trajectories of system $\Sigma$ consisting of $50$ cells.}}
		\label{st1}
	\end{center}
\end{figure}
\subsection{Fully Connected Network}\label{2}
In this example, we apply our results to an interconnected switched systems $\Sigma$ composed of $N\geq2$ linear switched subsystems $\Sigma_i, i\in [1;N],$ admitting multiple $\delta$-P storage functions and supply rates. In this respect, we choose the dynamics' parameters such that neither condition $\eqref{e:SFC11}$ nor $\eqref{e:SFC22}$ holds with common $\delta$-P storage functions and supply rates for all subsystems. In particular, as all subsystems are affine switched systems, we  choose their the dynamics' parameters such that the solution of the linear matrix inequality $\eqref{lmi}$ with common $Z_{i}$  and $Q_{i}$ (i.e. $Z_{ip_i}\!=\!Z_{ip'_i}$ and $Q_{ip_i}\!=\!Q_{ip'_i}$, $\forall p,p'\!\in\! P,i\!\in \![1;N]$) is infeasible. Hence, non of the subsystems admits a common $\delta$-P storage function and supply rate.
The dynamic of the
interconnected switched system $\Sigma$ has the set of modes $P\!\!=\!\!\{1,2\}^N, N\!\!\in\!\! \N_{\ge2}$, and it is given by 
	\begin{align*}\Sigma:\left\{
	\begin{array}[\relax]{rl}
	\mathbf{x}(k+1)&=A_{\mathsf{p}(k)}\mathbf{x}(k)+B_{\mathsf{p}(k)},\\
	\mathbf{y}(k)&=\mathbf{x}(k).
	\end{array}\right.
	\end{align*}The vector $B_p\in \R^{n}$, where $n=2N$, is defined as $\{B\}_{i1}=B_{p_i}$ such that $B_{p_i}=[-0.9;0.5]$ if $p_i=1$, and 
$B_{p_i}=[0.9;-0.2]$ if $p_i=2$,
$\forall i,j\in [1;N],i\neq j$. The elements of the matrix $A_p\in\R^{n\times n}$ are as follows:
	\begin{align*} 
	\{A\}_{ij}\!=\!\begin{bmatrix}
	0.015 \!\!\!&0  \\
	0\!\!\!& 0.015\\
	\end{bmatrix}\!\!,
	\{A\}_{ii}\!\!=\!\!A_{p_i}\!\!=\!\!\left\{
	\begin{array}{lr}
	\!\!\!\!\begin{bmatrix}
	0.05 \!\!\quad&0  \\
	0.9& 0.03\\
	\end{bmatrix}\quad\!\! \text{if}~ p_i=1,\vspace{0.1cm}\\
	\!\!\!\!\begin{bmatrix}
	0.02 &-1.2 \\
	0  &0.05\\
	\end{bmatrix}\quad \text{if}~  p_i=2.
	\end{array}\right.
	\end{align*}
Now, by introducing $\Sigma_i$ described by
	\begin{align*} 
	\Sigma_i:\left\{
	\begin{array}{rl}
	\mathbf{x}_i(k+1)&=A_{i\mathsf{p}_i(k)}\mathbf{x}_i(k)+\omega_i(k)+B_{i\mathsf{p}_i(k)},\\
	\mathbf{y}_{i1}(k)&=\mathbf{x}_i(k),\\
	\mathbf{y}_{i2}(k)&=\mathbf{x}_i(k),
	\end{array}\right.\\
	A_{i1}\!\!=\!\!\begin{bmatrix}
	0.05 \!\!\!\!\!\!\quad&0  \\
	0.9& 0.03\\
	\end{bmatrix}\!\!,
	A_{i2}\!\!=\!\!\begin{bmatrix}
	0.02 &-1.2 \\
	0  &0.05\\
	\end{bmatrix}\!\!,
	B_{i1}\!\!=\!\!\begin{bmatrix}
	-0.9  \\
	0.5\\
	\end{bmatrix}\!\!,
	B_{i2}\!\!=\!\!\begin{bmatrix}
	0.9 \\
	-0.2\\
	\end{bmatrix}\!\!,
	\end{align*}and the set of modes is $P_i\!\!=\!\!\{1,2\}$,
one can readily verify that $\Sigma\!=\!\mathcal{I}(\Sigma_1,\ldots,\Sigma_{N})$, where the elements of the coupling matrix $M$ are $\{M\}_{ii}\!\!=\!\!0_2$ and $\{M\}_{i,j}\!\!=\!\!\{A\}_{i,j}$, $\forall i,j\in\!\! [1;N],i\!\neq\! j$. 
Note that, for any $i\in[1;N]$,
conditions \eqref{e:SFC11} and \eqref{e:SFC22} are satisfied with $S_{ip_i}(x_i,\hat{x}_i)\!=\!{(x_i-\hat x_i)^TZ_{ip_i}(x_i-\hat x_i)}$,
	\begin{align*}
	Z_{i1}=\begin{bmatrix}
	0.3030&    0.0087\\
	0.0087&    0.4938
	\end{bmatrix},Z_{i2}=\begin{bmatrix}
	0.4899&   -0.0033\\
	-0.0033&    0.4291
	\end{bmatrix},
	\end{align*}$Q_{i1}=10^{-3}L_{i1},\kappa_{i1}\!=\!0.7,\underline{\alpha}_{i1}(s)\!=\!0.3s^2,,Q_{i2}=10^{-3}L_{i2},\kappa_{i2}\!=\!0.7,\underline{\alpha}_{i2}(s)\!=\!0.4s^2\!$, where
	\begin{align*}
	&L_{i1}\!\!=\!\!\!\begin{bmatrix}
	2.7\!&\!    0\!&\!   -1\!&\!   \!-3\!\\
	0\!&\!    1\!&\!  -3\!&\!  \!0\!\\
	-1\!&\!   -3\!&\!   -201.3\!&\!   \!-17\!\\
	-3\!&\!   0\!&\!   -1.7\!&\!    \!270.8\!
	\end{bmatrix}\!\!\!,L_{i2}\!\!=\!\!\!\begin{bmatrix}\!\!\!\!\!
	2.9\!&\!    0\!&\!   -1.4\!&\!    \!2.7\!\\
	0\!&\!    1.6\!&\!    2.7\!&\!   \!0\!\\
	-1.4\!&\!    2.7\!&\!    156\!&\!   \!17.5\!\\
	2.7\!&\!   0\!&\!    17.5\!&\!   \!-294\!
	\end{bmatrix}
	\end{align*}Since Assumption \ref{ass1} and $k_d\geq  \epsilon\frac{\ln (\mu)}{\ln (1/\kappa_p)}+1$ hold with $\mu=1.63$, $k_d=3$, $\epsilon=1.01$, one can easily find a matrix $\tilde{Q}$ such that $\forall q\in \{1,2\},\tilde{Q}-0.7^{\frac{-q}{\epsilon}}\sum_{p=1}^{2}Q_p\succeq0$ by using semi-definite programming such that function $\mathcal{V}_i((x_i,p_i,l_i),(\hat{x}_i,p_i,l_i))= {\sum_{i=1}^{N}S_{ip_i}(x_i,\hat{x}_i)}{\kappa^{{-l}/{\epsilon}}_{p_i}}$ is an augmented-storage function from $\hat{T}_i(\hat{\Sigma}_i)$ to $T_i(\Sigma_i)$. 
Choose an arbitrary $N$, then by choosing $\mu_1\!=\!\cdots\!=\!\mu_{N}=1$ and finite internal input sets $\hat{W}_i$ of $\hat{T}_i(\hat\Sigma_i)$ in such a way that $\prod_{i=1}^{N}\hat{W}_i=M\prod_{i=1}^{N}\hat{X}_i$, condition \eqref{e:MC1} and \eqref{e:MC3} are satisfied. Hence, using Theorem \ref{thm:3}, function $\tilde{\mathcal{S}}((x,p,l),(\hat{x},p,l))\!=\!\!\sum_{i=1}^{N}\!\mathcal{V}_i((x_i,p_i,l_i),(\hat{x}_i,p_i,l_i))$ is an alternating simulation function from $\hat{\mathcal{I}}(\hat{T}_1(\hat{\Sigma}_1),\ldots,\hat{T}_{N}(\hat{\Sigma}_{N}))$ to $\mathcal{I}(T_1(\Sigma_1),\ldots,T_{N}(\Sigma_{N}))$.

Given $N\geq5$, $X_i=[0,1]$, and $\eta_i=0.1$, we observe that constructing the symbolic model for the original system $\Sigma$ is only possible compositionally even with this small range of state set and coarse quantization parameters. The computation time for constructing symbolic models of $\Sigma_i$ is amounted to $0.53s$, using tool \texttt{SCOTS} \cite{Rungger} with the state quantization parameter $\eta_i=0.1$.
\section{Conclusion}
In this work, we proposed a compositional scheme for the construction of symbolic models of interconnected discrete-time switched systems. First, we used a notion of augmented-storage functions in order to construct compositionally an alternating simulation function that is used to quantify the error between the output behavior of the interconnected switched system and that of its abstraction. 
Furthermore, under some assumptions ensuring incremental passivity of each mode of switched subsystems, we showed how to construct symbolic models together with their corresponding augmented-storage functions of the concrete systems. 

\bibliographystyle{alpha}       

\bibliography{refcdc}

\newcommand{\etalchar}[1]{$^{#1}$}
\begin{thebibliography}{ZMEM{\etalchar{+}}14}

\bibitem[AMP16]{murat}
M.~Arcak, C.~Meissen, and A.~Packard.
\newblock {\em Networks of dissipative systems}.
\newblock SpringerBriefs in Electrical and Computer Engineering. Springer
  International Publishing, 2016.

\bibitem[BK08]{Katoen}
Christel Baier and Joost-Pieter Katoen.
\newblock {\em Principles of Model Checking (Representation and Mind Series)}.
\newblock The MIT Press, 2008.

\bibitem[CGG13]{Corronc}
E.~Le Corronc, A.~Girard, and G.~Goessler.
\newblock Mode sequences as symbolic states in abstractions of incrementally
  stable switched systems.
\newblock In {\em Proceedings of 52nd IEEE Conference on Decision and Control},
  pages 3225--3230, 2013.

\bibitem[dWOK12]{Kibangou}
Carlos~Canudas de~Wit, Luis~Leon Ojeda, and Alain~Y. Kibangou.
\newblock Graph constrained-ctm observer design for the grenoble south ring.
\newblock In {\em Proceedings of 13th IFAC Symposium on Control in
  Transportation Systems}, pages 197--202, 2012.

\bibitem[GGM16]{Gossler}
A.~Girard, G.~G{\"o}ssler, and S.~Mouelhi.
\newblock Safety controller synthesis for incrementally stable switched systems
  using multiscale symbolic models.
\newblock {\em IEEE Transactions on Automatic Control}, 61(6):1537--1549, 2016.

\bibitem[GP09]{Girard2009566}
A.~Girard and G.~J. Pappas.
\newblock Hierarchical control system design using approximate simulation.
\newblock {\em Automatica}, 45(2):566 -- 571, 2009.

\bibitem[GPT10]{Girard}
A.~Girard, G.~Pola, and P.~Tabuada.
\newblock Approximately bisimilar symbolic models for incrementally stable
  switched systems.
\newblock {\em IEEE Tranctions on Automatic Control}, 55(1):116--126, 2010.

\bibitem[HAT17]{omar}
O.~Hussein, A.~Ames, and P.~Tabuada.
\newblock Abstracting partially feedback linearizable systems compositionally.
\newblock {\em IEEE Control Systems Letters}, 1(2):227--232, 2017.

\bibitem[HSR98]{Rajamani}
T.~A. Henzinger, Q.~Shaz, and S.~K. Rajamani.
\newblock You assume, we guarantee: Methodology and case studies.
\newblock In {\em Proceedings of International Conference on Computer Aided
  Verification}, pages 440--451, 1998.

\bibitem[KAZ18]{Kim}
Eric~S. Kim, Murat Arcak, and Majid Zamani.
\newblock Constructing control system abstractions from modular components.
\newblock In {\em Proceedings of the 21st International Conference on Hybrid
  Systems: Computation and Control}, pages 137--146, 2018.

\bibitem[Kea11]{Michael}
Michael Keating.
\newblock {\em The Simple Art of SoC Design}.
\newblock Springer-Verlag, 2011.

\bibitem[Lib03]{liberzon}
Daniel Liberzon.
\newblock {\em Switching in Systems and Control}.
\newblock Birkh{\"a}user Basel, 2003.

\bibitem[Lof04]{YALMIP}
J.~Lofberg.
\newblock Yalmip : a toolbox for modeling and optimization in matlab.
\newblock In {\em Proceedings of the International Conference on Robotics and
  Automation}, pages 284--289, 2004.

\bibitem[MGW17]{meyer}
P.~J. Meyer, A.~Girard, and E.~Witrant.
\newblock Compositional abstraction and safety synthesis using overlapping
  symbolic models.
\newblock {\em IEEE Transactions on Automatic Control}, 63(6):1835--1841, 2017.

\bibitem[Mor96]{Morse}
A.~S. Morse.
\newblock Supervisory control of families of linear set-point controllers -
  part i. exact matching.
\newblock {\em IEEE Transactions on Automatic Control}, 41(10):1413--1431,
  1996.

\bibitem[MPS95]{MalerPnueliSifakis95}
O.~Maler, A.~Pnueli, and J.~Sifakis.
\newblock On the synthesis of discrete controllers for timed systems.
\newblock In {\em Proceedings of the 12th Symposium on Theoretical Aspects of
  Computer Science}, pages 229--242, 1995.

\bibitem[MSSM18]{Majumdar}
K.~Mallik, A-K Schmuck, S.~Soudjani, and R.~Majumdar.
\newblock Compositional synthesis of finite state abstractions.
\newblock {\em IEEE Transactions on Automatic Control}, 2018.

\bibitem[PPB16]{7403879}
G.~Pola, P.~Pepe, and M.~D.~Di Benedetto.
\newblock Symbolic models for networks of control systems.
\newblock {\em IEEE Transactions on Automatic Control}, 61(11):3663--3668,
  2016.

\bibitem[PPB18]{8115304}
G.~Pola, P.~Pepe, and M.~D.~D. Benedetto.
\newblock Decentralized supervisory control of networks of nonlinear control
  systems.
\newblock {\em IEEE Transactions on Automatic Control}, 63(9):2803--2817, Sept
  2018.

\bibitem[RZ16]{Rungger}
Matthias Rungger and Majid Zamani.
\newblock S{C}{O}{T}{S}: A tool for the synthesis of symbolic sontrollers.
\newblock In {\em Proceedings of the 19th International Conference on Hybrid
  Systems: Computation and Control}, pages 99--104, 2016.

\bibitem[SG17]{SAOUD}
Adnane Saoud and Antoine Girard.
\newblock Multirate symbolic models for incrementally stable switched systems.
\newblock In {\em Proceedings of 20th IFAC World Congress}, pages 9278 -- 9284,
  2017.

\bibitem[SGZ18]{arxiv}
A.~Swikir, A.~Girard, and M.~Zamani.
\newblock From dissipativity theory to compositional synthesis of symbolic
  models.
\newblock In {\em Proceedings of the 4th Indian Control Conference}, pages
  30--35, 2018.

\bibitem[SZ18]{swikir}
Abdalla Swikir and Majid Zamani.
\newblock Compositional synthesis of finite abstractions for networks of
  systems: {A} small-gain approach.
\newblock {\em CoRR}, abs/1805.06271, 2018.

\bibitem[Tab09]{Tabu}
Paulo Tabuada.
\newblock {\em Verification and Control of Hybrid Systems: A Symbolic
  Approach}.
\newblock Springer Publishing Company, Incorporated, 1st edition, 2009.

\bibitem[TI08]{Tazaki2008}
Y.~Tazaki and J.~I. Imura.
\newblock Bisimilar finite abstractions of interconnected systems.
\newblock In {\em Proceedings of the 11th International Conference on Hybrid
  Systems: Computation and Control}, pages 514--527, 2008.

\bibitem[VCL07]{VU}
L.~Vu, D.~Chatterjee, and D.~Liberzon.
\newblock Input-to-state stability of switched systems and switching adaptive
  control.
\newblock {\em Automatica}, 43(4):639 -- 646, 2007.

\bibitem[ZA17]{7857702}
M.~Zamani and M.~Arcak.
\newblock Compositional abstraction for networks of control systems: A
  dissipativity approach.
\newblock {\em IEEE Transactions on Control of Network Systems},
  5(3):1003--1015, 2017.

\bibitem[ZMEM{\etalchar{+}}14]{zamani2014symbolic}
M.~Zamani, P.~Mohajerin~Esfahani, R.~Majumdar, A.~Abate, and J.~Lygeros.
\newblock Symbolic control of stochastic systems via approximately bisimilar
  finite abstractions.
\newblock {\em IEEE Transactions on Automatic Control}, 59(12):3135--3150,
  2014.

\end{thebibliography}

\appendix

\end{document}